\documentclass{jpp}

\newcommand{\be}{\begin{equation}}
\newcommand{\ee}{\end{equation}}
\let\proof\relax 
\usepackage[psamsfonts]{amsfonts}
\usepackage{amsmath, amsthm, amssymb}
\usepackage{multirow}
\usepackage{enumerate}
\usepackage{setspace}
\usepackage[dvips]{graphicx}
\usepackage{natbib}

\newcommand{\derv}[1]{\frac{\partial}{\partial #1}}
\newcommand{\dtwo}[2]{\frac{\partial^2 #1}{\partial #2^2}}

\newcommand{\deriv}[2]{\frac{\partial #1}{\partial #2}}

\newcommand{\beqn}{\begin{equation}}
\newcommand{\eeqn}{\end{equation}}
\newcommand{\beqnar}{\begin{eqnarray}}
\newcommand{\eeqnar}{\end{eqnarray}}

\newtheorem{theorem}{Theorem}[section]
\newtheorem{proposition}[theorem]{Proposition}
\newenvironment{example}[1][Example]{\begin{trivlist}
\item[\hskip \labelsep {\bfseries #1}]}{\end{trivlist}}
\newenvironment{remark}[1][Remark]{\begin{trivlist}
\item[\hskip \labelsep {\bfseries #1}]}{\end{trivlist}}

\ifCUPmtlplainloaded \else
  \checkfont{eurm10}
  \iffontfound
    \IfFileExists{upmath.sty}
      {\typeout{^^JFound AMS Euler Roman fonts on the system,
                   using the 'upmath' package.^^J}%
       \usepackage{upmath}}
      {\typeout{^^JFound AMS Euler Roman fonts on the system, but you
                   dont seem to have the}%
       \typeout{'upmath' package installed. JPP.cls can take advantage
                 of these fonts, if you use 'upmath' package.^^J}%
      }
  \else
  \fi
\fi

\ifCUPmtlplainloaded \else
  \checkfont{msam10}
  \iffontfound
    \IfFileExists{amssymb.sty}
      {\typeout{^^JFound AMS Symbol fonts on the system, using the
                'amssymb' package.^^J}%
       \usepackage{amssymb}%
         \let\leq=\leqslant
         
      }{}
  \fi
\fi

\ifCUPmtlplainloaded \else
  \IfFileExists{amsbsy.sty}
    {\typeout{^^JFound the 'amsbsy' package on the system, using it.^^J}%
     \usepackage{amsbsy}}
    {\providecommand\boldsymbol[1]{\mbox{\boldmath $##1$}}}
\fi

\newsavebox{\astrutbox}
\sbox{\astrutbox}{\rule[-5pt]{0pt}{20pt}}

\title[Potential Vorticity in Magnetohydrodynamics]
{Potential Vorticity in Magnetohydrodynamics}

\author[G. M. Webb and R. L. Mace]%
{G.\ns M.\ns W\ls E\ls B\ls B$^1$%
 \thanks{Email address for correspondence: gmw0002@uah.edu}\ns
\and R.\ls L.\ns M\ls A\ls C\ls \ls E$^{2}$}

\affiliation{$^1$Center for Space Plasma and Aeronomic Research, 
The University of Alabama in Huntsville, 
Huntsville AL 35805, USA\\[\affilskip]
$^2$School of Chemistry and Physics, The University of Kwa-Zulu Natal,\\ 
Durban Westville, Durban, Natal 4000, South Africa}

\begin{document}


\maketitle

\begin{abstract}
A version of Noether's second theorem 
using Lagrange multipliers 
is used to investigate fluid relabelling symmetries 
conservation laws in magnetohydrodynamics (MHD).  
 We obtain a new generalized potential vorticity type conservation 
equation for MHD 
which takes into account entropy gradients and the ${\bf J}\times{\bf B}$
force on the plasma due to the current ${\bf J}$ and magnetic 
induction ${\bf B}$. This new conservation law for MHD is derived 
by using Noether's second theorem in conjunction with a class of fluid 
relabelling symmetries in which the symmetry generator for the Lagrange label
transformations is non-parallel to the magnetic field induction in Lagrange label space. This is associated with an Abelian Lie pseudo algebra and a foliated 
phase space in Lagrange label space. 
It contains as a special case Ertel's theorem in ideal fluid mechanics. 
 An independent derivation shows that
the new conservation law is also valid for more general physical situations. 
\end{abstract}


\maketitle


\section{Introduction}
Variational methods are widely used in physics, engineering and applied 
mathematics. Noether's theorems provide a route to deriving conservation 
laws for systems of differential  equations governed by an action principle.
Noether's theorem applies to systems of Euler-Lagrange equations 
that are in Kovalevskaya form (e.g \cite{Olver93}).
For other Euler-Lagrange systems, each nontrivial variational symmetry leads to a conservation law, 
but there is no guarantee that it is non-trivial.

The main aim of the present paper is to derive a generalized potential vorticity 
equation in MHD by using a version of Noether's second theorem  obtained by 
\cite{Hydon11}. The analysis  is related to recent work on MHD conservation 
laws by \cite{Webb14a, Webb14b}.  
{\bf  \cite{Cheviakov14} has derived new conservation laws for  fluid systems
involving vorticity and vorticity related equations  (potential type systems)
 including MHD and Maxwell's equations. One can use Cheviakov's scheme to derive the potential vorticity 
equation derived in the present analysis, which we indicate below (the 
actual potential vorticity equation we obtain is not explicitly 
given in Cheviakov (2014)). 

In ideal fluid mechanics, Ertel's theorem is usually written in the form:
\begin{equation}
\frac{d}{dt}\left(\frac{\boldsymbol{\omega}{\bf\cdot}\nabla\psi}{\rho}\right)
=0\quad\hbox{where}\quad \boldsymbol{\omega}=\nabla\times {\bf u}, \label{eq:1.1}
\end{equation}
is the fluid vorticity and $d/dt=\partial/\partial t+{\bf u}{\bf\cdot}\nabla$ is the Lagrangian time derivative 
moving with the flow, ${\bf u}$ is the fluid velocity and $\psi$ is a scalar advected with the 
flow, i.e. $d\psi/dt=0$. In addition it is required that $\nabla\psi\times\nabla S=0$ for (\ref{eq:1.1})
to apply. The quantity $I_e=\boldsymbol{\omega}{\bf\cdot}\nabla\psi/\rho$ is known as the potential vorticity or 
Ertel's invariant (e.g. \cite{Pedlosky87}; \cite{Salmon82, Salmon88}, \cite{Padhye96a, Padhye96b}, \cite{Padhye98}). 
The Eulerian version of (\ref{eq:1.1}) is:
\begin{equation}
\derv{t}\left(\boldsymbol{\omega}{\bf\cdot}\nabla\psi\right)
+\nabla{\bf\cdot}\left[(\boldsymbol{\omega}{\bf\cdot}\nabla\psi){\bf u}\right]=0. \label{eq:1.2}
\end{equation}
The generalized potential vorticity equation for MHD obtained in the present paper has the 
form:
\begin{equation}
\derv{t}\left(\boldsymbol{\omega}{\bf\cdot}\nabla\psi\right)
+\nabla{\bf\cdot}\left[(\boldsymbol{\omega}{\bf\cdot}\nabla\psi){\bf u} 
-\left(T\nabla S+\frac{{\bf J}\times{\bf B}}{\rho}\right)\times\nabla\psi\right]=0, 
\label{eq:1.3}
\end{equation}
where $T$ is the temperature of the gas, $S$ is the entropy and ${\bf J}=\nabla\times{\bf B}/\mu_0$ is the 
MHD current density.

Cheviakov's generalized vorticity conservation law applies to systems of partial differential equations of the 
form:
\begin{equation}
\nabla{\bf\cdot}{\bf N}=0,\quad \deriv{\bf N}{t}+\nabla\times{\bf M}=0. \label{eq:1.4}
\end{equation}
The system (\ref{eq:1.4}) has potential vorticity like conservation laws of the form:
\begin{equation}
\derv{t}\left({\bf N}{\bf\cdot}\nabla F\right)
+\nabla{\bf\cdot}\left({\bf M}\times\nabla F-F_t{\bf N}\right)=0, 
\label{eq:1.5}
\end{equation}
where $F(t,x,y,z)$ is an arbitrary function of $(t,x,y,z)$. The choices:
\begin{align}
{\bf N}=&\boldsymbol{\omega}=\nabla\times{\bf u}, \quad F=\psi({\bf x},t), \nonumber\\
{\bf M}=&-{\bf u}\times\boldsymbol{\omega}-\left(T\nabla S+\frac{{\bf J}\times{\bf B}}{\rho}\right), \label{eq:1.6}
\end{align}
in Cheviakov's conservation law (\ref{eq:1.5}) gives the potential vorticity type conservation law 
(\ref{eq:1.3}) obtained in our analysis, where $\psi$ is advected with the flow. The conservation law 
(\ref{eq:1.3}) was also given in a previous unpublished paper on the arxiv at http://arxiv.org/abs/1403.3133}. 
 
The emphasis of the present paper is on the connection between the new conservation law (\ref{eq:1.3}) 
 and non-field aligned fluid relabelling symmetries in MHD, and its derivation  
using Noether's second theorem. The work of \cite{Cheviakov14} on conservation 
laws for potential and vorticity equations is related to nonlocal
symmetries of partial differential equation (pde) systems (e.g. \cite{Bluman10}). 
\cite{Kelbin14} obtain new conservation laws for helically symmetric, plane, and rotationally 
symmetric flows.    
 \cite{Webb14a} derived MHD conservation laws 
using Lie dragging techniques (see also \cite{Tur93}).
\cite{Webb14b}, obtained advected 
invariants (i.e. advected invariant geometrical objects) in 
MHD and gas dynamics by using Noether's first and second theorems and 
the Euler Poincar\'e approach to variational problems in fluids and 
MHD (e.g. \cite{Holm98}, \cite{Cotter07}, \cite{Cotter13}).
These invariants are associated with the 
 fluid relabelling symmetries. \cite{Webb14c} 
developed a multi-symplectic formulation of ideal
 MHD and  gas dynamics equations based 
on Clebsch potentials (see also \cite{Cotter07} for related work). 

In this paper, we apply  Noether's second theorem 
(\cite{Hydon11}) to derive 
conservation laws for the equations. 
 \cite{Rosenhaus02} considers similar variational 
problems, and determines the effects of boundary conditions on Noether's 
second theorem. In particular, we derive a generalized potential vorticity 
conservation equation for MHD using the MHD fluid relabelling symmetries.

The MHD equations admit the ten-parameter Galilei Lie group. 
This includes the space and time translation symmetries, 
the space rotations and the 
Galilean boosts (e.g. \cite{Fuchs91}, \cite{Grundland95}, 
\cite{Webb07}). These symmetries give rise via Noether's first 
theorem to the (a) energy 
conservation law due to the time translation symmetry (b)\ the momentum 
conservation laws (space translation symmetries), (c)\  
angular momentum conservation laws (rotational symmetries) 
and (d) the center of mass conservation laws
(Galilean boosts symmetries). For a polytropic equation of state for the gas 
with $p\propto\rho^\gamma$ there are three extra scaling symmetries of the MHD 
and gas dynamic equations (\cite{Fuchs91}), which can be combined to give 
conservation laws (e.g. \cite{Webb07,Webb09b}). For $\gamma=(n+2)/n$ 
(n a positive integer) the equations also admit the projective symmetry 
(e.g. \cite{Ibragimov85}).  

In addition to the Galilei group there is a class of 
infinite dimensional fluid 
relabelling symmetries that leave the MHD equations invariant under 
transformation of the Lagrangian fluid labels. The fluid relabelling 
symmetries conservation laws are associated with Noether's second theorem 
(e.g. \cite{Salmon82,Salmon88}, \cite{Padhye96a,Padhye96b}; 
\cite{Padhye98}, 
 \cite{Ilgisnos00}, \cite{Kats03}, 
\cite{Webb05}, \cite{Webb07}, \cite{Cotter13}).

In Section 2, the MHD equations and the Lagrangian map are introduced. 
In Section 
3 we describe Noether's first and second theorems. We outline 
the formulation of Noether's second theorem using Lagrangian constraints, 
developed by \cite{Hydon11}.  
 In Section 4 we  formulate  Noether's theorems in MHD
using the Lagrangian MHD approach of \cite{Newcomb62}.  
We apply the Lagrange multiplier analysis of \cite{Hydon11}
to obtain Noether's second theorem in MHD, which is subsequently used 
to derive the new potential vorticity type conservation law 
for MHD. 
Section 5 concludes with a summary and discussion. 


\section{The MHD equations and the Lagrangian Map}

The time dependent MHD equations  consist of the mass, 
momentum and entropy advection equations:
\beqnar
&&{\deriv{\rho}{t}}+\nabla{\bf\cdot}(\rho{\bf u})=0, \label{eq:2.1}\\
&&\rho\left({\deriv{\bf u}{t}}+{\bf u\cdot\nabla u}\right)=-\nabla p
+{\bf J}\times {\bf B}-\rho\nabla\Phi, \label{eq:2.2}\\
&&{\deriv{S}{t}}+{\bf u}{\bf\cdot\nabla}S=0, \label{eq:2.3}
\eeqnar
coupled with Maxwell's equations in the MHD limit:
\beqnar
&&{\deriv{\bf B}{t}}=\nabla\times\left({\bf u}\times {\bf B}\right), 
\label{eq:2.4}\\
&& {\bf J}={\frac{\nabla\times {\bf B}}{\mu_0}},\quad \nabla{\bf\cdot} {\bf B}
=0. \label{eq:2.5}
\eeqnar
 Equations (\ref{eq:2.4})-(\ref{eq:2.5}) correspond to Faraday's
induction equation, Ampere's law for the current ${\bf J}$ and Gauss's 
equation $\nabla{\bf\cdot}{\bf B}=0$.  ${\bf F}_g=-\rho\nabla \Phi$ 
is the force due to an external gravitational potential 
$\Phi$ . The above equations are supplemented by an equation of 
state for the gas internal energy density $\varepsilon=\varepsilon(\rho,S)$. 
 For an ideal gas, the first law of 
thermodynamics gives:
\beqn
p=\rho{\deriv{\varepsilon}{\rho}}-\varepsilon,\quad \rho T
={\deriv{\varepsilon}{S}}, \label{eq:2.6}
\eeqn
for the pressure $p$ and the temperature $T$ of the gas.

\subsection{The Lagrangian map and variational formulation}

The Lagrangian map: ${\bf x}={\bf X}({\bf x}_0,t)$ is obtained by integrating 
the fluid velocity equation $d{\bf x}/dt={\bf u}({\bf x},t)$, subject to 
the initial condition ${\bf x}={\bf x}_0$ at time $t=0$. 
 In this approach, the 
mass continuity equation and entropy advection equation are replaced 
by the equivalent algebraic equations:
\beqn
\rho={\frac{\rho_0({\bf x}_0)}{J}}, \quad S=S({\bf x}_0), \label{eq:2.7}
\eeqn
where
\beqn
J=\det(x_{ij})\quad\hbox{and}\quad x_{ij}={\deriv{x^i({\bf x}_0,t)}{x_0^j}}. 
\label{eq:2.8}
\eeqn
Similarly, Faraday's equation (\ref{eq:2.4}) has the formal solution for the
magnetic field induction ${\bf B}$ of the form:
\beqn
B^i={\frac{x_{ij}B_0^j}{J}},\quad \nabla_0{\bf\cdot}{\bf B}_0=0. \label{eq:2.9}
\eeqn
The solution (\ref{eq:2.9}) for $B^i$ is equivalent to the frozen in 
field theorem in MHD (e.g. \cite{Stern66}, \cite{Parker79}), and the 
initial condition $\nabla_0{\bf\cdot}{\bf B}_0=0$
is imposed in order to ensure that Gauss's law $\nabla{\bf\cdot}{\bf B}=0$ 
is satisfied. 



The action for the MHD system is:
\beqn
A=\int \int \ell\ d^3x dt\equiv \int\int \ell^0\  d^3x_0 dt, 
\label{eq:2.15} 
\eeqn
where
\beqn
\ell={\frac{1}{2}}\rho |{\bf u}|^2-\varepsilon(\rho,S)-{\frac{B^2}{2\mu}}
-\rho\Phi,\quad \ell^0=\ell J, \label{eq:2.16}
\eeqn
are the Eulerian ($\ell$) and Lagrangian ($\ell^0$) Lagrange densities
respectively. Using (\ref{eq:2.7})-(\ref{eq:2.9}), and (\ref{eq:2.16}) 
we obtain:
\beqn
\ell^0={\frac{1}{2}}\rho_0 |{\bf x}_t|^2
-J\varepsilon\left({\frac{\rho_0}{J}},S\right)
 -{\frac{x_{ij}x_{is} B_0^j B_0^s}{2\mu J}}-\rho_0\Phi, \label{eq:2.17}
\eeqn
for $\ell^0$. Note that in the Lagrange density 
$\ell^0=\ell^0({\bf x}_0,t; {\bf x}, {\bf x}_t, x_{ij})$, 
${\bf x}_0$ and $t$ are the independent variables, and ${\bf x}$ 
and its derivatives with respect to ${\bf x}_0$ and $t$ are dependent
variables.

Stationary point conditions for the action (\ref{eq:2.15})
gives the Euler-Lagrange equations:
\beqn
{\frac{\delta A}{\delta x^i}}={\deriv{\ell^0}{x^i}}
-{\derv{t}}\left({\deriv{\ell^0}{x^i_t}}\right)
-{\derv{x_0^j}}\left({\deriv{\ell^0}{x_{ij}}}\right)=0,
 \label{eq:2.18}
\eeqn
where $x_{ij}\equiv \partial x^i/\partial x_0^j$. 
Evaluation of the variational derivative (\ref{eq:2.18}) 
gives the Lagrangian momentum equation for the system in 
the form (\cite{Newcomb62}):
\beqn
\rho_0\left({\dtwo{x^i}{t}}+{\deriv{\Phi}{x^i}}\right)
+{\derv{x_0^j}}\left\{ A_{kj} \left[
\left(p+{\frac{B^2}{2\mu}}\right)\delta^{ik}-{\frac{B^i B^k}{\mu_0}}\right]
\right\}=0, \label{eq:2.19}
\eeqn
where $A_{kj}=\hbox{cofac}(x_{kj})$. Dividing (\ref{eq:2.19}) 
by $J$, and using the fact that $\partial A_{kj}/\partial x_0^j=0$, 
gives the Eulerian form of the momentum equation (\ref{eq:2.2}).

Equation (\ref{eq:2.19}) can be reduced to a system of three coupled 
nonlinear wave equations for ${\bf x}={\bf x}({\bf x}_0,t)$ in which 
$\rho({\bf x}_0)$, $S({\bf x}_0)$, and ${\bf B}_0({\bf x}_0)$ are 
given functions of ${\bf x}_0$ (see e.g. \cite{Webb07}). 

\section{Noether's Theorems}

Noether's theorems describes the relationship between conservation laws, 
and Lie symmetries of differential equation systems represented by 
a variational principle. In Section 3.1 we outline Noether's first theorem, which applies 
to conservation laws associated with a variational principle where the action is invariant 
under a finite Lie group. Section 3.2 describes Noether's second theorem. 
Noether's second theorem applies when the action is invariant under an 
infinite dimensional pseudo-Lie group, in which case there are in general dependencies between 
the Euler Lagrange equations for the system.  
%


\subsection{Noether's First Theorem}

Noether's first theorem concerns the form of a conservation law 
for a system of partial differential equations described by an 
action principle for the case where there are a finite number of 
symmetries, or for a finite 
Lie group of symmetries that leave the action invariant under variational or 
divergence symmetries. Noether's second theorem applies when the action is 
invariant under an infinite dimensional Lie pseudo group, in which case there 
are relationships between the Euler Lagrange equations that must be 
taken into account (e.g. \cite{Noether18}, \cite{Padhye96a,Padhye96b}, 
\cite{Hydon11}, \cite{Cotter13}). In this case,
the generators of the Lie pseudo algebra depend on arbitrary functions of the
independent variables. 

In general, the problem concerns the variational and divergence symmetries 
of the action
\begin{equation}
J=\int L({\bf x},[{\bf u}])\ d{\bf x}, \label{eq:C1}
\end{equation}
where $[{\bf u}]$ denotes the dependent variables and their derivatives 
with respect to the independent variables ${\bf x}$. In Section 4, we use ${\bf u}$ to 
denote the fluid velocity, which  should
not cause any confusion because of the different context.
 The Euler Lagrange equations describing the equation 
system of interest are obtained by varying the action $J$:
\begin{equation}
\delta_u J([{\bf u}])= \int \deriv{L}{u^\alpha_{,I}} \delta u^\alpha_{,I}\
 d {\bf x}=\int \deriv{L}{u^\alpha_{,I}} {\bf D}_I(\delta u^\alpha)\ d{\bf x}
=\int{\bf E}_\alpha (L) \delta u^\alpha\ d{\bf x}=0, \label{eq:C2}
\end{equation}
where we have dropped  divergence surface terms, that are assumed to vanish 
on the boundary $\partial R$ of the integration region $R$ of interest. 
 The operator $E_\alpha(L)$ in (\ref{eq:C2})
is known as the Euler operator (e.g. \cite{Bluman89}).
Here we use the usual multi-index notation where 
$u^\alpha_{,I}
=\partial u^\alpha/\partial x^{i_1}\partial x^{i_2}\ldots \partial x^{i_n}$, 
i.e. 
$I=i_1i_2i_3\ldots i_n$ for some arbitrary given $n$. 
From (\ref{eq:C2}) we obtain the Euler-Lagrange equations for the system:
\begin{equation}
E_\alpha (L)\equiv (-D)_I \left(\deriv{L}{u^\alpha_{,I}}\right)=0, \quad 
\alpha=1,2,\ldots ,q, \label{eq:C3}
\end{equation}
which are equivalent to the differential equation system of interest. 
In (\ref{eq:C3}) for $I=i_1i_2i_3\ldots j_k$, $(-D)_I$ is defined as:
\begin{equation}
(-D)_I=(-1)^k D_I=(-D_{i_1})(-D_{i_2})\ldots (-D_{i_k}), \label{eq:C3a}
\end{equation}
(see \cite{Olver93}, p. 245). Here $D_i$ is the total partial derivative 
with respect to $x^i$. 
 
In the general case, we are interested in variations that correspond to 
infinitesimal  Lie 
transformations of the form:
\begin{equation}
x'^i=x^i+\epsilon V^{x^i}, \quad u'^\alpha=u^\alpha+\epsilon V^{u^\alpha}, 
\label{eq:C4}
\end{equation}
and under the divergence transformation:
\begin{equation}
L'=L+\epsilon D_i\Lambda^i+O(\epsilon^2). \label{eq:C5}
\end{equation}
The $\Lambda^i$ are gauge potentials. In (\ref{eq:C4}) $V^{x^i}$ and the $V^{u^\alpha}$ 
are the infinitesimal generators of the Lie group transformation used in the analysis.
It is well known 
(e.g. \cite{Bluman89}, \cite{Olver93}) 
that the Euler-Lagrange equations $E_\alpha(L)=0$ 
obtained by searching for stationary point solutions of $\delta_{\bf u}J$ in 
(\ref{eq:C2}) remain invariant (i.e. are the same) under a divergence 
transformation of the form (\ref{eq:C5}). 
If $\Lambda^i=0$, the infinitesimal 
transformations leaving the action invariant are referred to as variational 
symmetries, but if $\Lambda^i\neq 0$ the transformations are referred to as 
divergence transformations (e.g. \cite{Bluman89}). 
The transformations 
(\ref{eq:C4}) are equivalent to the characteristic or canonical Lie 
transformations:
\begin{equation}
{\bf x}'={\bf x}, \quad u'^\alpha=u^\alpha+\epsilon {\hat V}^{u^\alpha}, 
\quad  \hat{V}^{u^\alpha}=V^{u^\alpha}-V^{x^i} u^\alpha_{,i}, 
\label{eq:C6}
\end{equation}
in which the independent variables ${\bf x}$ are fixed (e.g. 
\cite{Ibragimov85},
\cite{Bluman89}, \cite{Olver93}). 

The variation of the action under transformations 
(\ref{eq:C4})-(\ref{eq:C6}) is:
\begin{align}
\delta J=&\lim_{\epsilon\to 0} \frac{1}{\epsilon} \left(
\int_{{\cal R}'} L'({\bf x}',[{\bf u}'])\ d{\bf x}'
-\int_{\cal R} L({\bf x},[{\bf u}])\ d{\bf x}\right)\nonumber\\
=&\int_{\cal R} \left(\tilde{X}L+LD_i V^{x^i}+D_i\Lambda^i\right)
\ d{\bf x}\nonumber\\
\equiv&\int_{\cal R} \left[\hat{X}L+D_i \left(L V^{x^i}+\Lambda^i\right)\right]
\ d{\bf x}. \label{eq:C7}
\end{align}
Here
\begin{equation}
\tilde{X}=\hat{X}+V^{x^i}D_i, \label{eq:C8}
\end{equation}
is the extended Lie transformation operator corresponding to  
 (\ref{eq:C4})  and
\begin{equation}
\hat{X}=\hat{V}^{u^\alpha}\derv{u^\alpha}+D_i\left(\hat{V}^{u^\alpha}\right) 
\derv{u^\alpha_i}+\ldots 
\equiv D_I\left(\hat{V}^{u^\alpha}\right) 
\derv{u^\alpha_I} \label{eq:C9}
\end{equation}
is the extended operator corresponding to the  
characteristic Lie transformations (\ref{eq:C6}). 
From (\ref{eq:C7}) $\delta J=0$ if: 
\begin{equation}
\hat{X}L+D_i\left(L V^{x^i}+\Lambda^i\right)=0. 
\label{eq:C10}
\end{equation}

Using integration by parts to calculate $\hat{X}L$ gives:
\begin{equation}
\hat{X}(L)=\hat{V}^{u^\alpha}E_\alpha(L)
+D_i\left(W^i\left[{\bf u},\hat{V}^{\bf u}\right]\right), \label{eq:C11}
\end{equation}
where
\begin{align}
W^i[{\bf u},{\bf v}]=&v^\gamma\frac{\delta L}{\delta u^\gamma_i}
+v^\gamma_j\frac{\delta L}{\delta u^\gamma_{ij}}
+v^\gamma_{jk}\frac{\delta L}{\delta u^\gamma_{ijk}} +\ldots, \nonumber\\
\frac{\delta L}{\delta\psi}=&\deriv{L}{\psi}-D_i\left(\deriv{L}{\psi_i}\right)
+D_i D_j\left(\deriv{L}{\psi_{ij}}\right)-\ldots. \label{eq:C12}
\end{align}
In (\ref{eq:C12}) $\psi$ refers to the dependent variables $u^\alpha$ 
or any of their derivatives of any order 
The expression for  $\delta L/\delta\psi$ in (\ref{eq:C12}) 
is given by \cite{Ibragimov07,Ibragimov11}  in calculating the surface 
term $W^i[{\bf u},{\bf v}]$. 
Using (\ref{eq:C11}) in (\ref{eq:C10}) the 
Lie invariance condition (\ref{eq:C10}) reduces to:
\begin{equation}
\hat{V}^{u^\alpha}E_\alpha(L)+D_i\left(W^i+L V^{x^i} +\Lambda^i\right)=0. 
\label{eq:C13}
\end{equation}
For a finite number of Lie symmetries, the 
$\hat{V}^{u^\alpha}$ and Euler Lagrange equations 
$E_\alpha(L)=0$ ($\alpha=1,2.\ldots q$), 
are independent. (\ref{eq:C13}) then reduces to the conservation law:
\begin{equation}
D_i\left(W^i+L V^{x^i} +\Lambda^i\right)=0. 
\label{eq:C14}
\end{equation}
Equation (\ref{eq:C14}) is the conservation law for 
Noether's first theorem. 

\subsection{Noether's Second Theorem}

Noether's second theorem (\cite{Olver93}, \cite{Hydon11}) 
expresses the idea that there must exist a relation between the Euler-Lagrange 
equations if the symmetry operator 
$\hat{V}^{\bf u}({\bf x},[{\bf u},{\bf g}])$
depends on an arbitrary function ${\bf g}({\bf x})$. 

\leftline{\it Theorem:\ \cite{Olver93}}
The variational problem (\ref{eq:C2}) admits an infinite 
dimensional group of variational symmetries 
whose characteristics ${\hat V}^{\bf u}({\bf x},[{\bf u},{\bf g}])$ 
depend on an arbitrary function ${\bf g}({\bf x})$ (and its derivatives) 
if and only if there exist differential operators ${\cal D}^1, {\cal D}^2, 
\ldots {\cal D}^q$, not all zero, such that:
\begin{equation}
{\cal D}^1E_1(L)+{\cal D}^2E_2(L)+\ldots +{\cal D}^qE_q(L)=0, \label{eq:C15}
\end{equation}
for all ${\bf x}$ and ${\bf u}$. 

A proof of the theorem is given by \cite{Olver93}. 
\cite{Hydon11}  
 give a simpler proof of the {\it only if} part of the proof, 
(based on \cite{Noether18})
and identify the operators $\left\{{\cal D}^s:\ 1\leq s\leq q\right\}$.

 \cite{Hydon11}, 
apply the operator 
$E_g$ to (\ref{eq:C13}) to obtain:
\begin{equation}
E_g\left\{\hat{V}^{u^\alpha}({\bf x},[{\bf u};g]) 
E_\alpha (L)\right\}=0, \label{eq:C16}
\end{equation}
as the required differential relation (\ref{eq:C15}) between 
the Euler-Lagrange equations. The theorem extends immediately to 
variational symmetries whose characteristics ${\hat V}^{u^\alpha}$ 
depend on $R$ independent arbitrary functions 
${\bf g}=(g^1({\bf x}),\ldots, g^R({\bf x}))$ and their derivatives. 
This gives $R$ differential relations between the Euler Lagrange equations:
\begin{equation}
E_{g^r}\left\{\hat{V}^{u^\alpha}({\bf x},[{\bf u};{\bf g}]) E_\alpha(L)\right\}
=(-D)_I\left(\deriv{\hat{V}^{u^\alpha}({\bf x},[{\bf u};{\bf g}])}{g^r_{,I}}
E_\alpha(L)\right)=0, \quad r=1,\ldots, R. \label{eq:C17}
\end{equation}
It is useful to consider (\ref{eq:C16}) as Euler Lagrange equations for 
the action
\begin{equation}
\hat{J}[{\bf u};{\bf g}]=\int\hat{L}({\bf x};[{\bf u};{\bf g}])\ d{\bf x}, 
\label{eq:C18}
\end{equation}
where
\begin{equation}
\hat{L}({\bf x},[{\bf u};{\bf g}])=\hat{V}^{u^\alpha}({\bf x},[{\bf u};g]) 
E_\alpha (L({\bf x},[{\bf u}])). \label{eq:C19}
\end{equation}

If the functions ${\bf g}=(g^1,g^2,\ldots,g^R)$ 
are subject to $S$ constraints, of the form:
\begin{equation}
{\cal D}_{sr}(g^r)=0,\quad s=1,\ldots, S, \label{eq:C20}
\end{equation}
where the ${\cal D}_{sr}$ are differential operators, then the constraints can 
be incorporated in the Lagrangian $\hat{L}$, by using the modified Lagrangian:
\begin{equation}
\hat{L}({\bf x},[{\bf u};{\bf g}])=\hat{V}^{u^\alpha}({\bf x},[{\bf u};{\bf g}]) 
E_\alpha (L({\bf x},[{\bf u}]))- \nu^s {\cal D}_{sr}(g^r), \label{eq:C21}
\end{equation}
in which the $\nu^s$ are Lagrange multipliers.  

By varying the action (\ref{eq:C18}) and (\ref{eq:C21}) with respect 
to $g^r$, we obtain
\begin{equation}
\delta\left(\int\nu^s{\cal D}_{sr}(g^r)\ d{\bf x}\right)
=\langle\nu^s,{\cal D}_{sr}(\delta g^r)\rangle
= \langle {\cal D}_{sr}^\dagger(\nu^s),\delta g^r\rangle, 
\label{eq:C22}
\end{equation}
where the surface terms have been dropped. 
The angle brackets define the usual inner product for functions. 
Taking variations of (\ref{eq:C18}) and (\ref{eq:C21}) with respect 
to $g^r$ yields
\begin{equation}
\frac{\delta {\hat J}}{\delta g^r}=(-D)_I 
\left(\deriv{{\hat V}^{u^\alpha}({\bf x},[{\bf u},{\bf g}])}{g^r_{,I}} 
E_\alpha(L)\right)-{\cal D}_{sr}^\dagger(\nu^s)=0,\quad  r=1,
\ldots, R. \label{eq:C23}
\end{equation}
Thus, if $S<R$, it may be possible to eliminate the Lagrange multipliers 
in (\ref{eq:C23}). In any event, (\ref{eq:C23}) relates the Lagrange multipliers $\nu^s$ to the solutions of the original Euler Lagrange 
equations (\ref{eq:C3}).

\section{MHD Conservation Laws and Symmetries}
We consider Noether's theorems in Lagrangian MHD 
(e.g. \cite{Webb05} and \cite{Webb07}), 
which may be used to derive 
MHD conservation laws.  
   
\begin{proposition}[Noether's theorem]\label{prop4.1}
If the action (\ref{eq:2.15}) is invariant to $O(\epsilon)$ under the 
infinitesimal Lie transformations:
\beqn
x'^i=x^i+\epsilon V^{x^i}, \quad
x'^j_0=x^j_0+\epsilon V^{x_0^j}, \quad
t'=t+\epsilon V^t, \label{eq:4.20}
\eeqn
and the divergence transformation:
\beqn
\ell^{0'}=\ell^0+\epsilon D_\alpha \Lambda_0^\alpha+O(\epsilon^2), 
\label{eq:4.21}
\eeqn
(here $D_0\equiv \partial/\partial t$ and $D_i\equiv \partial/\partial x_0^i$
are the total derivative operators in the jet-space consisting 
of the derivatives of $x^k({\bf x}_0,t)$ and physical quantities 
that depend on ${\bf x}_0$ and $t$) then the MHD system admits the 
Lagrangian identity:
\beqn
\hat{V}^{x^k} E_{x^k}\left(\ell^0\right)
+{\deriv{I^0}{t}}+{\deriv{I^j}{x_0^j}}=0, \label{eq:4.22}
\eeqn
where
\beqnar
&&I^0=\rho_0 u^k {\hat V}^{x^k}+V^t \ell^0+\Lambda_0^0, 
\label{eq:4.23}\\
&&I^j={\hat V}^{x^k} \left[\left( p+{\frac{B^2}{2\mu}}\right) \delta^{ks}
-{\frac{B^k B^s}{\mu}}\right] A_{sj} + V^{x_0^j} \ell^0 +\Lambda_0^j,
\label{eq:4.24}
\eeqnar
In (\ref{eq:4.23})-(\ref{eq:4.24}) we use the notation: 
$x_{kj}=\partial x^k/\partial x_0^j$, $A_{sj}=cofac(x_{sj})$ is the 
co-factor matrix corresponding to 
$x_{sj}$ and
\beqn
{\hat V}^{x^k({\bf x}_0,t)} = V^{x^k({\bf x}_0,t)}-\left(V^t{\derv{t}}
+V^{x_0^s} {\derv{x_0^s}} \right)x^k({\bf x}_0,t), \label{eq:4.25}
\eeqn
is the canonical Lie symmetry transformation generator corresponding
to the Lie transformation (\ref{eq:4.20})
(i.e. $x'^k=x^k+\epsilon {\hat V}^{x^k}$, $t'=t$, $x'^j_0=x_0^j$). 
If the Euler Lagrange equations $E_{x^k}(\ell^0)=0$ are independent 
one obtains the Lagrangian conservation law:
\beqn
{\deriv{I^0}{t}}+{\deriv{I^j}{x_0^j}}=0, \label{eq:4.25a}
\eeqn
The conservation law (\ref{eq:4.25a}) is satisfied for the case of a finite 
group of transformations but is not the only possible solution 
of the variational problem for the case of an infinite Lie pseudo group, 
because in that case, the Euler Lagrange equations are not all independent.
This latter case corresponds to Noether's second theorem.  
\end{proposition}
\begin{proof}
Using Noether's theorem (e.g. \cite{Bluman89}) we obtain:
\beqnar
&&I^0=W^0+V^t\ell^0+\Lambda_0^0\equiv {\deriv{\ell^0}{x^k_t}}
{\hat V}^{x^k}+V^t\ell^0+\Lambda_0^0, \nonumber\\
&&I^j=W^j+\ell^0 V^{x_0^j}+\Lambda_0^j\equiv {\deriv{\ell^0}{x_{kj}}}
{\hat V}^{x^k}+\ell^0 V^{x_0^j}+\Lambda_0^j, \label{eq:4.26}
\eeqnar
 for the conserved density $I^0$ and flux components $I^j$.
Using (\ref{eq:2.17}) for $\ell^0$ in (\ref{eq:4.26})
to evaluate the derivatives of $\ell^0$ with respect to $x^k_t$ and 
$x_{ij}$ gives the expressions (\ref{eq:4.23})-(\ref{eq:4.24})
for $I^0$ and $I^j$. 
 Proofs of Noether's first theorem are 
given in \cite{Bluman89} and \cite{Olver93} (see \cite{Webb05} 
 for the MHD system, 
including fully nonlinear waves). 
\end{proof}

\begin{remark}
The condition for the action (\ref{eq:2.15}) to be invariant to $O(\epsilon)$ 
under the divergence transformation of the form (\ref{eq:4.20})-(\ref{eq:4.21})
is:
\beqn
{\tilde X} \ell^0+\ell^0 \left(D_t V^t+D_{x_0^j} V^{x_0^j}\right) 
+D_t \Lambda_0^0+D_{x_0^j} \Lambda_0^j=0, \label{eq:4.27}
\eeqn
where
\beqn
{\tilde X}=V^t {\derv{t}}+V^{x^k}{\derv{x^k}}+V^{x^k_t} {\derv{x^k_t}}
+V^{x_{kj}} {\derv{x_{kj}}}+\cdots, \label{eq:4.28}
\eeqn
is the extended Lie transformation operator acting on the jet space 
of the Lie transformation (\ref{eq:4.20}). 
\end{remark}
\begin{proposition}\label{prop4.2}
The Lagrangian conservation law (\ref{eq:4.25a}) can be written as an 
Eulerian conservation law of the form (\cite{Padhye98}):
\beqn
{\deriv{F^0}{t}}+{\deriv{F^j}{x^j}}=0, \label{eq:euler1}
\eeqn
where
\beqn
{F^0}={\frac{I^0}{J}},\quad F^j={\frac{u^j I^0+x_{jk}I^k}{J}}, 
\quad (j=1,2,3), \label{eq:euler2}
\eeqn
are the conserved density $F^0$ and flux components $F^j$, $J=\det(x_{ij})$ 
and $x_{ij}=\partial x^i/\partial x^j_0$.
\end{proposition}

\begin{proposition}\label{prop4.3}

The Lagrangian conservation law (\ref{eq:4.25a}) with conserved 
density $I^0$ of (\ref{eq:4.23}), and flux $I^j$ of (\ref{eq:4.24}), 
is equivalent to the 
Eulerian conservation law:
\beqn
{\deriv{F^0}{t}}+{\deriv{F^j}{x^j}}=0, \label{eq:euler3}
\eeqn
where
\beqnar
&&F^0=\rho u^k {\hat V}^{x^k({\bf x}_0,t)} +V^t \ell+\Lambda^0, 
\label{eq:euler4}\\ 
&&F^j={\hat V}^{x^k({\bf x}_0,t)}\left(T^{jk}-\ell\delta^{jk}\right)
+V^{x^j}\ell+\Lambda^j, \label{eq:euler5}\\
&&T^{jk}=\rho u^j u^k +\left(p+{\frac{B^2}{2\mu}}\right)\delta^{jk}
-{\frac{B^jB^k}{\mu}}, \label{eq:euler6}\\
&&\Lambda^0={\frac{\Lambda_0^0}{J}},
\quad \Lambda^j={\frac{u^j \Lambda_0^0+x_{js}\Lambda_0^s}{J}}. 
\label{eq:euler7}
\eeqnar
In (\ref{eq:euler3})-(\ref{eq:euler7}) $T^{jk}$ is the Eulerian momentum 
flux tensor (i.e. the spatial components of the stress energy tensor) 
and ${\hat V}^{x^k({\bf x}_0,t)}$ is the canonical symmetry generator 
(\ref{eq:4.25}). 
\end{proposition}
\begin{remark}
For a pure fluid relabelling symmetry $V^{\bf x}=V^t=0$, and 
Proposition \ref{prop4.3} gives:
\begin{align}
F^0=&\hat{V}^{\bf x}{\bf\cdot}(\rho {\bf u})+\Lambda^0, 
\label{eq:euler8}\\
{\bf F}=&\hat{V}^{\bf x}{\bf\cdot}\left[\rho {\bf u}\otimes{\bf u}
+\left(\varepsilon+p+\rho\Phi +\frac{B^2}{\mu_0}
-\frac{1}{2}\rho u^2\right)\sf{I}
-\frac{{\bf B}\otimes{\bf B}}{\mu_0}\right] +\boldsymbol{\Lambda}, 
\label{eq:euler9}
\end{align}
for the conserved density $F^0$ and flux ${\bf F}$ where 
$\boldsymbol{\Lambda}=(0,\Lambda^1,\Lambda^2,\Lambda^3)$. 
Note that  
$V^{x_0}$, $V^{\bf x}$ and $V^t$ must satisfy the Lie invariance condition 
(\ref{eq:4.27}). If the symmetry generators depend on arbitrary functions of 
 $x_0$, $t$ or $x$, then 
 Noether's second theorem applies. 
\end{remark}

\subsection{Fluid Relabelling Symmetries} 

Consider infinitesimal Lie transformations of the form 
(\ref{eq:4.20})-(\ref{eq:4.21}), with
\beqn
V^t=0,\quad V^{\bf x}=0,\quad V^{{\bf x}_0}\neq 0,  
\label{eq:4.31}
\eeqn
which leave the action (\ref{eq:2.15}) invariant. The extended Lie 
transformation operator ${\tilde X}$ for the case (\ref{eq:4.31}) 
has generators:
\beqnar
&&{\hat V}^{\bf x}=-V^{{\bf x}_0}{\bf\cdot}\nabla_0 {\bf x},\quad 
V^{{\bf x}_t}=-D_t \left(V^{{\bf x}_0}\right) {\bf\cdot}\nabla_0 {\bf x},
\nonumber\\
&&V^{\nabla_0{\bf x}}=-\nabla_0\left(V^{{\bf x}_0}\right)
{\bf\cdot}\nabla_0 {\bf x}. 
\label{eq:4.32}
\eeqnar
The condition (\ref{eq:4.27}) for a divergence symmetry of the action 
reduces to:
\begin{align}
&\nabla_0{\bf\cdot}\left( \rho_0 V^{{\bf x}_0}\right) 
\left( {\frac{1}{2}}|{\bf u}|^2-\Phi ({\bf x})
-{\frac{\varepsilon+p}{\rho}}\right)
 -J{\deriv{\varepsilon(\rho,S)}{S}}
V^{{\bf x}_0}{\bf\cdot}\nabla_0 S\nonumber\\
&-D_t\left(\rho_0 V^{{\bf x}_0}\right){\bf u}{\bf\cdot}
\left(\nabla_0{\bf x}\right)^T
\nonumber\\
&+{\frac{1}{\mu J}}(\nabla_0 {\bf x}){\bf\cdot}(\nabla_0{\bf x})^T
{\bf :}{\bf B}_0\bigl[\nabla_0\times(V^{{\bf x}_0}\times {\bf B}_0)
-\nabla{\bf\cdot}{\bf B}_0 V^{{\bf x}_0}\bigr]\nonumber\\
&=-\partial\Lambda_0^\alpha/\partial x_0^\alpha.
\label{eq:4.33}
\end{align}
A relatively simple class of solutions of the Lie determining 
equations (\ref{eq:4.33}) for the fluid relabelling symmetries are 
obtained by solving the equations:
\beqnar
&&\nabla_0{\bf\cdot}\left( \rho_0 V^{{\bf x}_0}\right)=0,\quad 
V^{{\bf x}_0}{\bf\cdot}\nabla_0 S=0,
\quad D_t\left(\rho_0 V^{{\bf x}_0}\right)=0, \nonumber\\
&&\nabla_0\times\left(V^{{\bf x}_0}\times {\bf B}_0\right)=0,\quad 
\nabla_0{\bf\cdot}{\bf B}_0=0,\quad \Lambda^\alpha_0=0. \label{eq:4.34}
\eeqnar
Equations (\ref{eq:4.34}) are the Lie determining equations for the fluid 
relabelling symmetries obtained by \cite{Padhye98} and \cite{Webb05}.
The condition $\nabla_0{\bf\cdot}{\bf B}_0=0$ is required 
in order that $\nabla{\bf\cdot}{\bf B}=0$. 
$\nabla_0{\bf\cdot}{\bf B}_0\neq 0$ solutions are possible for  
(\ref{eq:4.33}), but these are not physically relevant. The equations 
for $\nabla{\bf\cdot}{\bf B}\neq 0$ are relevant to numerical MHD codes 
where numerical errors in $\nabla{\bf\cdot}{\bf B}\neq 0$ occur
(e.g. \cite{Webb09a} provide a simple MHD wave solution with 
$\nabla{\bf\cdot}{\bf B}\neq 0$, which conserves total momentum 
and energy).  
\cite{Morrison80,Morrison82a,Morrison82} provides a non-canonical Poisson 
bracket for MHD both for $\nabla{\bf\cdot}{\bf B}=0$ and for  $\nabla{\bf\cdot}{\bf B}\neq 0$ 
(see also \cite{Holm83a, Holm83b}, for noncanonical Poisson brackets using 
the magnetic vector potential).   

\begin{proposition}\label{prop4.4}
The condition (\ref{eq:4.33}) for a divergence symmetry of the action 
converted to Eulerian form is:
\begin{align}
&\nabla{\bf\cdot}\left(\rho\hat{V}^{\bf x}\right) 
\left(h+\Phi({\bf x})-\frac{1}{2}|{\bf u}|^2\right) 
+\rho T\hat{V}^{\bf x}{\bf\cdot}\nabla S
+\rho {\bf u}{\bf\cdot}\left(\frac{d\hat{V}^{\bf x}}{dt}
-\hat{V}^{\bf x}{\bf\cdot}\nabla{\bf u}\right)\nonumber\\
&+\frac{\bf B}{\mu_0}{\bf\cdot}\left[-\nabla\times
\left(\hat{V}^{\bf x}\times{\bf B}\right) 
+\hat{V}^{\bf x}\nabla{\bf\cdot}{\bf B}\right]
=-\nabla_\alpha\Lambda^\alpha, 
\label{eq:4.35aa}
\end{align}
where 
\beqn
\nabla_\alpha\Lambda^\alpha=\deriv{\Lambda^0}{t}+\deriv{\Lambda^i}{x^i}, 
\label{eq:4.35ba}
\eeqn
is the  4-divergence of  
$\boldsymbol{\Lambda}= \left(\Lambda^0,\Lambda^1,\Lambda^2,\Lambda^3\right)$. 
 $\boldsymbol{\Lambda}$ is related the the Lagrange 
label space vector $\Lambda_0^\alpha$ by the transformations:
\beqn
\Lambda^\alpha=\Lambda_0^\beta B_{\beta\alpha}\equiv \Lambda^\beta_0 
\frac{x_{\alpha \beta}}{J}, \label{eq:4.35ca}
\eeqn
where $x_{\alpha\beta}=\partial x^\alpha/\partial x_0^\beta$, $J=\det(x_{ij})$ 
and $B_{\alpha\beta}=\hbox{cofac}(\partial x_0^\alpha/\partial x^\beta)$ (the 
transformations (\ref{eq:4.35ca}) are the same as those in (\ref{eq:euler7}); 
note that $\alpha$, $\beta$ have values $0,1,2,3$).
\end{proposition}

\begin{remark}
The invariance condition (\ref{eq:4.35aa}) combined with the  
momentum equation (\ref{eq:2.2}) and 
 mass continuity equation (\ref{eq:2.1}) 
formally implies the conservation law:
\begin{equation}
\derv{t}\left(\rho{\bf u}{\bf\cdot}\hat{V}^{\bf x}+\Lambda^0\right)
+\nabla{\bf\cdot}\left[\rho\hat{V}^{\bf x}
\left(h+\Phi-\frac{1}{2}u^2\right) +\frac{\tilde{\bf E}\times{\bf B}}{\mu_0}
+\boldsymbol{\Lambda}\right]=0, \label{eq:4.35da}
\end{equation}
where
\begin{equation}
\tilde{\bf E}=-\hat{V}^{\bf x}\times {\bf B}, \label{eq:4.35ea}
\end{equation}
is a pseudo electric field associated with  
 $\hat{V}^{\bf x}$. If one takes $\hat{V}^{\bf x}$ to have the 
dimensions of velocity, then 
${\bf S}_{\tilde{\bf E}}=\tilde{\bf E}\times{\bf B}/\mu_0$ 
is analogous to the Poynting flux. 
The conservation law (\ref{eq:4.35da}) is equivalent to 
the conservation law $\partial F^0/\partial t+\nabla{\bf\cdot} {\bf F}=0$, 
where $F^0$ and ${\bf F}$ are given by (\ref{eq:euler8}) and (\ref{eq:euler9}).
\end{remark}
\begin{proof}
 The proof 
 of (\ref{eq:4.35aa}) is given in \cite{Webb14b}.
In Appendix A, (\ref{eq:4.35aa}) is derived by using Lie dragging methods 
 (e.g. \cite{Holm98},
\cite{Cotter13} and \cite{Webb14a}, 
\cite{Ivey03}, \cite{Harrison71}, \cite{Flanders63}). 
\end{proof}

A more general class of solutions of (\ref{eq:4.33}) than (\ref{eq:4.34}) satisfy the equation 
system: 
\beqnar
&&\nabla_0{\bf\cdot}\left( \rho_0 V^{{\bf x}_0}\right)=0,\quad
V^{{\bf x}_0}{\bf\cdot}\nabla_0 S=0,
\quad u^i x_{ij}D_t\left(\rho_0 V^{x_0^j}\right)=0, \nonumber\\
&&\frac{x_{ia}x_{ib}}{J} B_0^b 
\nabla_0\times\left(V^{{\bf x}_0}\times {\bf B}_0\right)^a=0,\quad
\nabla_0{\bf\cdot}{\bf B}_0=0. \label{eq:4.35}
\eeqnar
By using ${\bf x}$, $t$ as the independent variables 
(i.e. ${\bf x}_0={\bf x}_0({\bf x},t)$), the Lie determining equations 
(\ref{eq:4.35}) can be written in the form:
\begin{align}
&\nabla{\bf\cdot}(\rho {\hat V}^{\bf x})=0, 
\quad {\hat V}^{\bf x}{\bf\cdot}\nabla S=0, \label{eq:4.35a}\\
&\rho_0{\bf u}{\bf\cdot}
\left(\frac{d{\hat V}^{\bf x}}{d\tau} 
-{\hat V}^{\bf x}{\bf\cdot}\nabla {\bf u}\right)=0, \label{eq:4.35b}\\
&{\bf B}{\bf\cdot}\nabla\times\left({\hat V}^{\bf x}\times {\bf B}\right)=0. 
\label{eq:4.35c}
\end{align}

\subsection{The Lagrange multiplier approach}

Following \cite{Hydon11} we regard the fluid 
relabelling symmetry equations as auxiliary constraints. Thus, for 
 Noether's second theorem, we consider  
 the action  (\ref{eq:C18}) with Lagrange density:
\begin{equation}
\hat{J}[{\bf x};{\bf g}]=\int\hat{L}({\bf y};[{\bf x};{\bf g}])\ d{\bf y}, 
\label{eq:4.36}
\end{equation}
where ${\bf x}=(x,y,z)$ is the Eulerian position coordinate, ${\bf y}=(t,x_0,y_0.z_0)$ are the independent variables, and ${\bf g}$ corresponds to a relabelling symmetry generated by $V^{{\bf x}_0}$. The Lagrangian:
$\hat{L}$ has the form:
\begin{align}
\hat{L}({\bf y},[{\bf x};{\bf g}])=&-V^{x_0^j}x_{ij} E_{x^i}\left(\ell^0\right)
-\nu^1 \nabla_0{\bf\cdot}\left( \rho_0 V^{{\bf x}_0}\right)
-\nu^2 V^{{\bf x}_0}{\bf\cdot}\nabla_0 S\nonumber\\
&-\nu^3_a  
\nabla_0\times\left(V^{{\bf x}_0}\times {\bf B}_0\right)^a
-\nu^4_a D_t\left(\rho_0 V^{x_0^a}\right), \label{eq:4.37}
\end{align}
where $\nu^1$, $\nu^2$, $\nu^3_a$ and $\nu^4_a$ are Lagrange multipliers. 

Setting the variational derivative 
$\delta\hat{J}/\delta V^{x_0^a}=0$ we obtain the relations
between the Euler-Lagrange equations  of the form:
\beqn
\frac{\delta{\hat J}}{\delta V^{x^a_0}}=-x_{ia} E_{x^i}\left(\ell^0\right)
+\rho_0\deriv{\nu^1}{x_0^a}
-\nu^2\deriv{S}{x_0^a}+B_0^j
\left(\deriv{\nu_a^3}{x_0^j}
-\deriv{\nu_j^3}{x_0^a}\right)
+\rho_0\deriv{\nu_a^4}{\tau}=0, \label{eq:4.38}
\eeqn
where $\partial/\partial\tau=\partial_t+{\bf u}{\bf\cdot}\nabla$ 
is the Lagrangian time derivative moving with the flow. 

Next note from (\ref{eq:2.18}):
\beqn
E_{x^i}\left(\ell^0\right)
=-\left(\rho_0\deriv{u^i}{\tau}
+\derv{x_0^j}\left\{ A_{kj} \left[
\left(p+{\frac{B^2}{2\mu}}\right)\delta^{ik}-{\frac{B^i B^k}{\mu_0}}\right]
\right\}\right), \label{eq:4.38a}
\eeqn
(we neglect gravity, i.e. we set $\Phi=0$).  Using the first law of 
thermodynamics, we obtain:
\beqn
T\nabla_0 S=\nabla_0 h-\frac{1}{\rho}\nabla_0p, \label{eq:4.38b}
\eeqn
where $T$ is the temperature and $h=(\varepsilon+p)/\rho$ is the enthalpy 
of the gas. Using (\ref{eq:4.38a}) and (\ref{eq:4.38b}) 
we find solutions of (\ref{eq:4.38}) for 
$\nu^1$, $\nu^2$, $\boldsymbol{\nu}^3$ and $\boldsymbol{\nu}^4$ given by:
\begin{align}
&\nu^1=\frac{1}{2} u^2-h, \quad h=\frac{\varepsilon+p}{\rho}, 
\quad  \nu^2=-\rho_0 T, \nonumber\\
&\boldsymbol{\nu}^3=\frac{\nabla_0 {\bf x}{\bf\cdot}{\bf B}}{\mu}, 
\quad \boldsymbol{\nu}^4=-(\nabla_0{\bf x}){\bf\cdot}{\bf u}. 
\label{eq:4.38c}
\end{align}

\subsection{Non-field aligned $V^{{\bf x}_0}$}

\cite{Webb05} investigated solutions of the fluid relabelling symmetry 
determining equations (\ref{eq:4.34}) for the cases: (a)\ $V^{{\bf x}_0}$ parallel 
to ${\bf B}_0$  and (b)\  $V^{{\bf x}_0}$ not parallel to ${\bf B}_0$. 
In case (b) equations (\ref{eq:4.34}) imply the existence of a 
potential $\psi ({\bf x}_0)$ 
such that:
\begin{equation}
V^{{\bf x}_0}\times {\bf B}_0=\nabla_0\psi, 
\quad \nabla_0{\bf\cdot} \left(\rho_0 
V^{{\bf x}_0}\right)=0,\quad V^{{\bf x}_0}{\bf\cdot}\nabla_0 S=0,\quad 
\nabla_0{\bf\cdot}{\bf B}_0=0. \label{eq:nfa1}
\end{equation}
Note that (\ref{eq:nfa1}) satisfies the Lie determining equation:
\begin{equation}
\nabla_0\times\left(V^{{\bf x}_0}\times {\bf B}_0\right)=0. \label{eq:nfa2}
\end{equation}
Using $\nabla_0{\bf\cdot} \left(\rho_0V^{{\bf x}_0}\right)=0$, (\ref{eq:nfa2})
can be written in the form:
\begin{equation}
\nabla_0\times\left(V^{{\bf x}_0}\times {\bf B}_0\right)
=\rho_0\left[{\bf b}_0,V^{{\bf x}_0}\right]=0, \label{eq:nfa3}
\end{equation}
where ${\bf b}_0={\bf B}_0/\rho_0$ and 
$\left[{\bf b}_0,V^{{\bf x}_0}\right]^i\partial/ \partial x_0^i$ is the 
Lie bracket of the tangent vector fields:
\begin{equation}
{\bf b}_0=b_0^i \derv{x_0^i},\quad V^{{\bf x}_0}=V^{x_0^i} \derv{x_0^i}. 
\label{eq:nfa4}
\end{equation}
From (\ref{eq:nfa3}), the vector fields 
$\{{\bf b}_0,V^{{\bf x}_0}\}$ describes a two dimensional  Abelian Lie 
algebra. From Frobenius theorem (e.g. \cite{Olver93}, p. 39)),  
the integral trajectories of ${\bf b}_0$ and $V^{{\bf x}_0}$ generate the 
the integral foliation $\psi({\bf x}_0)=c$ ($c=const.$). A parametric 
representation of the surface $\psi=c$ is ${\bf x}_0={\bf x}_0(\phi,\chi,c)$ 
where 
\begin{equation}
V^{{\bf x}_0}{\bf\cdot}\nabla_0=\derv{\phi}\quad \hbox{and}\quad 
{\bf b}_0{\bf\cdot}\nabla_0=\derv{\chi}\label{eq:nfa5}
\end{equation}
are directional derivatives (tangent vectors) in the surface $\psi=c$. 
We set:
\begin{align}
&\omega^1=\nabla_0\phi,\quad \omega^2=\nabla_0\chi,\quad \omega^3=\nabla_0\psi, 
\nonumber\\
&{\bf e}_1=\deriv{{\bf x}_0}{\phi},\quad {\bf e}_2=\deriv{{\bf x}_0}{\chi}, \quad 
{\bf e}_3=\deriv{{\bf x}_0}{\psi}. \label{eq:nfa6}
\end{align}
The bases $\{{\bf e}_1,{\bf e}_2,{\bf e}_3\}$ and 
$\{\omega^1,\omega^2,\omega^3\}$ are dual bases, i.e. 
$\langle\omega^i,{\bf e}_j\rangle=\delta^i_j$. Set $g_{ij}={\bf e}_i{\bf\cdot} 
{\bf e}_j$ and $g^{ij}=\omega^i{\bf\cdot}\omega^j$ as the metric tensors. 
 The bases are related 
by the equations: 
\begin{align}
&\omega^i=g^{ij} {\bf e}_j,\quad {\bf e}_i=g_{ij} \omega^j, \nonumber\\
&{\bf e}_a\times {\bf e}_b=\sqrt{|g|} \varepsilon_{abc} \omega^c, \quad 
\omega^a\times\omega^b=\frac{\varepsilon_{abc}}{\sqrt{|g|}} {\bf e}_c, 
\label{eq:nfa7}
\end{align} 
where $\varepsilon_{abc}$ is the anti-symmetric Levi-Civita tensor density
in three dimensions. 
The Lie determining equations (\ref{eq:4.34}) or (\ref{eq:nfa1}) have 
solutions:
\begin{align}
&\rho_0V^{{\bf x}_0}=\omega^2\times \omega^3=\nabla_0\chi\times \nabla_0\psi
=\nabla_0\times(\chi\nabla_0\psi), \nonumber\\
&{\bf B}_0=\omega^3\times\omega^1=\nabla_0\psi\times\nabla_0\phi
=\nabla_0\times(\psi\nabla_0\phi), \nonumber\\
&\rho_0=\frac{1}{\sqrt{|g|}}=\sqrt{|G|}, \quad S=S(\chi,\psi), \label{eq:nfa8}
\end{align}
where $g=\det(g_{ij})$ and $G=\det (g^{ij})$ are the determinants of the 
metric tensors $g_{ij}$ and $g^{ij}$ respectively. The potentials $\phi$ $\psi$ and $\chi$ are functions only of ${\bf x}_0$, i.e. they are scalar
invariants advected with the flow. 

The Eulerian symmetry generator $\hat{V}^{\bf x}$ and ${\bf B}$  
 are given by:
\begin{equation}
\hat{V}^{\bf x}=\frac{\nabla\psi\times\nabla\chi}{\rho},\quad 
{\bf B}=\nabla\psi\times\nabla\phi. \label{eq:nfa9}
\end{equation}
Because $\chi$, $\phi$ and $\psi$ are functions only of the Lagrange labels 
${\bf x}_0$ it follows that:
\begin{equation}
\frac{d\phi}{dt}=\frac{d\chi}{dt}=\frac{d\psi}{dt}=0. \label{eq:nfa10}
\end{equation}
 Note that 
$\psi$ and $\phi$ are Euler potentials for the magnetic field ${\bf B}$. 

The basic Noether identity (\ref{eq:4.22}) , namely:
\beqn
\hat{V}^{x^k} E_{x^k}\left(\ell^0\right)
+{\deriv{I^0}{t}}+\nabla_0{\bf\cdot}{\bf I}=0, \label{eq:nfa11}
\eeqn
 describes the conservation laws for the symmetries
(\ref{eq:nfa8})-(\ref{eq:nfa9}), for which
\begin{align}
&I^0=(\nabla_0\psi\times\nabla_0\chi){\bf\cdot}\nabla_0{\bf x}{\bf\cdot}{\bf u}, \label{eq:nfa12}\\
&{\bf I}=(\nabla_0\psi\times\nabla_0\chi){\bf\cdot}
\left\{\left[\frac{p+B^2/(2\mu_0)-\ell}{\rho}\right]{\sf I}
-\frac{(\nabla_0{\bf x}{\bf\cdot}{\bf B}){\bf B}_0}
{\mu_0\rho_0}\right\}, \label{eq:nfa13}
\end{align}
where ${\sf I}$ is the unit $3\times 3$ dyadic. 

\cite{Webb05}  used (\ref{eq:nfa11}) and Noether's second theorem 
to study the conservation 
law associated with the non-field aligned solutions 
(\ref{eq:nfa8})-(\ref{eq:nfa9}) of the fluid relabelling 
Lie determining equations. Unfortunately, due to the dropping of the 
$\nabla_0{\bf\cdot} {\bf I}$ term when integrating (\ref{eq:nfa11}) 
over $d^3{\bf x}^0$, there is an error in their analysis. Below we 
obtain the correct conservation law associated with the  
symmetry (\ref{eq:nfa8})-(\ref{eq:nfa9}) where $\chi({\bf x}_0)$ is an 
arbitrary function defining a pseudo Lie algebra, by using the 
Hydon and Mansfield  approach.

\begin{proposition}\label{prop4.6}
The fluid relabelling symmetry (\ref{eq:nfa8}) with generators:
\begin{align}
&V^{{\bf x}_0}=\frac{\nabla_0\chi\times\nabla_0\psi}{\rho_0},
\quad \hat{V}^{\bf x}=\frac{\nabla\psi\times\nabla\chi}{\rho}, \nonumber\\
&{\bf B}_0=\nabla_0\psi\times\nabla_0\phi,
\quad {\bf B}=\nabla\psi\times\nabla\phi, \nonumber\\
 &\rho_0=\frac{1}{\sqrt{|g}|},\quad S=S(\chi,\psi), \label{eq:nfa14}
\end{align}
used with Noether's second theorem gives rise to the generalized 
Bianchi identity:
\begin{equation}
\frac{d}{dt}\left(\rho_0
\frac{\boldsymbol{\omega}{\bf\cdot}\nabla\psi}{\rho}\right) 
-\nabla_0{\bf\cdot} \left({\bf G}\times \nabla_0\psi\right)
+\nabla_0\psi{\bf\cdot} \nabla_0\times
\left(\frac{E_{\bf x}(\ell^0){\bf\cdot}(\nabla_0{\bf x})^T}{\rho_0}\right)=0, 
\label{eq:nfa15}
\end{equation}
where $\boldsymbol{\omega}=\nabla\times{\bf u}$ is the fluid vorticity and 
\begin{equation}
{\bf G}={\bf F}{\bf\cdot}\left(\nabla_0 {\bf x}\right)^T, \quad
{\bf F}=T\nabla S+\nabla\left(\frac{1}{2} u^2-h\right)
+\frac{{\bf J}\times{\bf B}}{\rho}, \label{eq:nfa16}
\end{equation}
and ${\bf J}=\nabla\times{\bf B}/\mu_0$ is the MHD current density. 
The Eulerian form of (\ref{eq:nfa15}) 
is:
\begin{equation}
\derv{t} \left(\boldsymbol{\omega}{\bf\cdot}\nabla\psi\right) 
+\nabla{\bf\cdot}\left[(\boldsymbol{\omega}{\bf\cdot}\nabla\psi){\bf u}
-{\bf F}\times \nabla\psi\right] 
+\nabla\psi{\bf\cdot}\nabla\times
\left(\frac{E_{\bf x}\left(\ell^0\right)}{\rho_0}\right)=0.
\label{eq:nfa17}
\end{equation}
For the cases where the Euler Lagrange equations are satisfied, i.e. 
$E_{\bf x}\left(\ell^0\right)=0$, (\ref{eq:nfa15}) and (\ref{eq:nfa17}) 
reduce to the conservation laws:
\begin{align}
&\frac{d}{dt}\left(\rho_0
\frac{\boldsymbol{\omega}{\bf\cdot}\nabla\psi}{\rho}\right)
-\nabla_0{\bf\cdot} \left({\bf G}\times \nabla_0\psi\right)=0, 
\label{eq:nfa18}\\
&\derv{t} \left(\boldsymbol{\omega}{\bf\cdot}\nabla\psi\right) 
+\nabla{\bf\cdot}\left[(\boldsymbol{\omega}{\bf\cdot}\nabla\psi){\bf u}
-{\bf F}\times \nabla\psi\right]=0. \label{eq:nfa19}
\end{align}
Equation (\ref{eq:nfa19}) can be written in the form:
\begin{equation}
\derv{t} \left(\boldsymbol{\omega}{\bf\cdot}\nabla\psi\right)
+\nabla{\bf\cdot}\left[(\boldsymbol{\omega}{\bf\cdot}\nabla\psi){\bf u}
-\left(T\nabla S+\frac{{\bf J}\times{\bf B}}{\rho}\right)
\times\nabla\psi\right]=0. \label{eq:nfa20}
\end{equation}
\end{proposition}

\begin{remark}
More generally,  conservation laws  
 (\ref{eq:nfa18}) and (\ref{eq:nfa19}) hold for the MHD equations, provided   
 $d\psi/dt=(\partial/\partial t+{\bf u}{\bf\cdot}\nabla)\psi=0$.
\end{remark}

\begin{proof}
First note that: 
\begin{equation}
\delta\hat{J}=\int\frac{\delta{\hat J}}{\delta V^{{\bf x}_0}}{\bf\cdot} 
\delta V^{{\bf x}_0}\ 
d^3 {\bf x}_0 dt=\int\frac{\delta \hat{J}}{\delta\chi}\delta\chi 
d^3{\bf x}_0 dt. \label{eq:nfa21}
\end{equation}
Using $V^{{\bf x}_0}$ from (\ref{eq:nfa8}), i.e. 
$V^{{\bf x}_0}=\nabla_0\chi\times\nabla_0\psi/\rho_0$, integrating by parts, 
and dropping the surface term , we obtain the generalized Bianchi 
identity:
\begin{equation}
\frac{\delta\hat{J}}{\delta\chi}= \nabla_0{\bf\cdot}
\left[\left(\frac{1}{\rho_0} 
\frac{\delta{\hat J}}{\delta V^{{\bf x}_0}}\right)\times\nabla_0\psi\right]
=0. \label{eq:nfa23}
\end{equation}
Using the expression (\ref{eq:4.38}) for $\delta\hat{J}/\delta V^{{\bf x}_0}$ 
in (\ref{eq:nfa23}) we obtain:
\begin{equation}
\frac{\delta\hat{J}}{\delta\chi}=\nabla_0\psi{\bf\cdot}\nabla_0\times \left(
-\frac{E_{\bf x}\left(\ell^0\right){\bf\cdot}(\nabla_0 {\bf x})^T}{\rho_0}
+\deriv{\boldsymbol{\nu}^4}{\tau}+{\bf G}\right), \label{eq:nfa24}
\end{equation}
where $\partial/\partial \tau\equiv d/dt$ is the Lagrangian time derivative 
following the flow and 
\begin{equation}
{\bf G}=\nabla_0\nu^1-\frac{\nu^2}{\rho_0}\nabla_0 S+{\bf b}_0{\bf\cdot}
\left[\nabla_0\boldsymbol{\nu}^3-(\nabla_0\boldsymbol{\nu}^3)^T\right], 
\label{eq:nfa25}
\end{equation}
and ${\bf b}_0={\bf B}_0/\rho_0$. 

Using the vector identity
\begin{equation}
\nabla_0{\bf\cdot}({\bf A}\times{\bf B})
=(\nabla_0\times{\bf A}){\bf\cdot}{\bf B}
-(\nabla_0\times{\bf B}){\bf\cdot}{\bf A}, \label{eq:nfa26}
\end{equation}
and (\ref{eq:4.38c}) for $\nu^1$, $\nu^2$, $\boldsymbol{\nu}^3$ 
and $\boldsymbol{\nu}^4$, (\ref{eq:nfa24}) can be written in the form:
\begin{equation}
\frac{\delta\hat{J}}{\delta\chi}=\nabla_0\psi{\bf\cdot}\nabla_0\times\left(
-\frac{E_{\bf x}\left(\ell^0\right){\bf\cdot}(\nabla_0 {\bf x})^T}{\rho_0}
\right)
+\deriv{Q}{\tau}+\nabla_0{\bf\cdot}({\bf G}\times\nabla_0\psi), 
\label{eq:nfa27}
\end{equation}
where 
\begin{equation}
Q=\nabla_0\psi{\bf\cdot}\nabla_0\times\boldsymbol{\nu}^4 
=-\nabla_0\psi{\bf\cdot}\nabla_0 \times
\left[{\bf u}{\bf\cdot}(\nabla_0 {\bf x})^T\right]. 
\label{eq:nfa28}
\end{equation}
By converting from ${\bf x}_0$ to ${\bf x}$ as independent variables  
 in (\ref{eq:nfa28}) gives:
\begin{equation}
Q=-\varepsilon_{ijk} \deriv{\psi}{x_0^i} \deriv{u^s}{x_0^j} x_{sk}
\equiv -\frac{(\boldsymbol{\omega}{\bf\cdot}
\nabla\psi)\rho_0}{\rho}, \label{eq:nfa29}
\end{equation}
where $\boldsymbol{\omega}=\nabla\times{\bf u}$ is the fluid vorticity and 
$x_{sk}=\partial x^s/\partial x_0^k$. 

Similarly, using (\ref{eq:4.38c}) for $\nu^1$, $\nu^2$, $\boldsymbol{\nu}^3$, 
the expression for ${\bf G}$ in (\ref{eq:nfa25}) reduces to (\ref{eq:nfa16}).
Using (\ref{eq:nfa16}) for ${\bf G}$ and (\ref{eq:nfa29}) for $Q$ gives:
\begin{align}
\frac{\delta\hat{J}}{\delta\chi}=&-\biggl\{ \frac{d}{dt}\left(\rho_0 
\frac{\boldsymbol{\omega}{\bf\cdot}\nabla\psi}{\rho}\right)
-\nabla_0{\bf\cdot}\left({\bf G}\times\nabla_0\psi\right)\nonumber\\
&+\nabla_0\psi{\bf\cdot}\nabla_0\times
\left(
\frac{E_{\bf x}\left(\ell^0\right){\bf\cdot}(\nabla_0 {\bf x})^T}{\rho_0}
\right)\biggr\}=0. \label{eq:nfa30}
\end{align}
Equation (\ref{eq:nfa30}) is equivalent to (\ref{eq:nfa15}). 
Converting (\ref{eq:nfa30}) to its Eulerian form gives (\ref{eq:nfa17}). 
The conservation laws (\ref{eq:nfa18}) and (\ref{eq:nfa19}) follow 
from (\ref{eq:nfa15}) and (\ref{eq:nfa17}) and by setting 
$E_{\bf x}(\ell^0)=0$. This completes the proof.
\end{proof}

Below we provide an independent verification of the conservation law 
(\ref{eq:nfa19}) or (\ref{eq:nfa20}). This derivation shows that the potential
$\psi$ only needs to be an advected invariant. Thus, for example, $\psi$ 
could be any of the advected scalar Lie dragged invariants, e.g.:
\begin{equation}
\psi_1=\frac{{\bf B}{\bf\cdot}\nabla S}{\rho},
\quad \psi_2=\frac{{\bf A}{\bf\cdot}{\bf B}}{\rho},
\quad \psi_3=\frac{{\bf B}{\bf\cdot}\nabla}{\rho}\left(\frac{{\bf A}{\bf\cdot}{\bf B}}{\rho}\right), \label{eq:nfa31}
\end{equation}
could be used for $\psi$. In the above examples, we choose 
the gauge for ${\bf A}$ so that ${\bf A}{\bf\cdot}d{\bf x}$ 
is a Lie dragged 1-form 
(see e.g. \cite{Webb14a} for more detail).

To derive (\ref{eq:nfa19}) or (\ref{eq:nfa20}) write the MHD momentum 
equation in the form:
\begin{equation}
\frac{d}{dt}({\bf u}{\bf\cdot}d{\bf x})
=\left[{\bf u}_t-{\bf u}\times\boldsymbol{\omega} 
+\nabla|{\bf u}|^2\right]{\bf\cdot}d{\bf x}={\bf F}{\bf\cdot}d{\bf x},
\label{eq:nfa32}
\end{equation}
where ${\bf F}$ is given by (\ref{eq:nfa16}). 

A more recognizable form of (\ref{eq:nfa32}) is:
\begin{equation}
{\bf u}_t-{\bf u}\times\boldsymbol{\omega} 
+\nabla|{\bf u}|^2=T\nabla S+\nabla\left(\frac{1}{2}u^2-h\right)
+\frac{{\bf J}\times {\bf B}}{\rho}. \label{eq:nfa34}
\end{equation}
Taking the curl of (\ref{eq:nfa34}) gives the fluid vorticity equation for MHD 
in the form:
\begin{equation}
\boldsymbol{\omega}_t-\nabla\times({\bf u}\times\boldsymbol{\omega})
=\nabla\times {\bf F}. \label{eq:nfa35}
\end{equation}
For $\psi$ a scalar  advected invariant $\nabla\psi$ satisfies the equation 
\begin{equation}
\derv{t}\nabla\psi+\nabla({\bf u}{\bf\cdot}\nabla\psi)\equiv 
\nabla\left(\deriv{\psi}{t}+{\bf u}{\bf\cdot}\nabla\psi\right)=0. 
\label{eq:nfa36}
\end{equation}
In the Lie dragging formalism, $\alpha=d\psi=\nabla\psi{\bf\cdot}d{\bf x}$ 
is a Lie dragged 1-form that is invariant moving with the flow (e.g. 
\cite{Tur93}, 
 \cite{Webb14a}). 

Taking the scalar product of (\ref{eq:nfa35}) with $\nabla\psi$ and 
the scalar product of (\ref{eq:nfa36}) with $\boldsymbol{\omega}$ and adding 
the resultant equations gives:
\begin{equation}
\derv{t}\left(\boldsymbol{\omega}{\bf\cdot}\nabla\psi\right) -\nabla\psi{\bf\cdot}\nabla\times({\bf u}\times\boldsymbol{\omega}) 
+(\nabla\times{\bf u}){\bf\cdot}\nabla({\bf u}{\bf\cdot}\nabla\psi)
-\nabla\psi{\bf\cdot}\nabla\times {\bf F}=0. \label{eq:nfa37}
\end{equation}
Equation (\ref{eq:nfa37}) can be re-written in the conservation form:
\begin{equation}
\derv{t}\left(\boldsymbol{\omega}{\bf\cdot}\nabla\psi\right)
+\nabla{\bf\cdot}\left[{\bf u}\times\nabla({\bf u}{\bf\cdot}\nabla\psi)
-({\bf u}\times\boldsymbol{\omega}) \times\nabla\psi
- {\bf F}\times\nabla\psi\right]=0. \label{eq:nfa38}
\end{equation}
Using the identity
$\nabla\psi\times({\bf u}\times\boldsymbol{\omega})
=(\nabla\psi{\bf\cdot}\boldsymbol{\omega}){\bf u}
-(\nabla\psi{\bf\cdot}{\bf u})\boldsymbol{\omega}$,
(\ref{eq:nfa38}) reduces to:
\begin{equation}
\derv{t}\left(\boldsymbol{\omega}{\bf\cdot}\nabla\psi\right)
+\nabla{\bf\cdot}\left[(\boldsymbol{\omega}{\bf\cdot}\nabla\psi){\bf u}
-\nabla\times[({\bf u}{\bf\cdot}\nabla\psi){\bf u}] 
-{\bf F}\times\nabla\psi\right]
=0. \label{eq:nfa40}
\end{equation}
Because the divergence of a curl is zero, (\ref{eq:nfa40}) is equivalent to 
the conservation law (\ref{eq:nfa19}) or (\ref{eq:nfa20}). This completes
our independent derivation of (\ref{eq:nfa19}).

\subsubsection{Eulerian Approach}
Proposition \ref{prop4.6} used the Lagrangian version of Noether's 
second theorem to derive the conservation laws (\ref{eq:nfa15}) 
and (\ref{eq:nfa17}). An alternative approach is to use an Eulerian version
of the variational principle (\ref{eq:4.36})-(\ref{eq:4.37}), which we 
outline below. 

From (\ref{eq:4.36}) we obtain:
\begin{equation}
\hat{J}=\int \hat{L}\ d^3{\bf x}_0 dt=\int\tilde{L}\ d^3{\bf x} dt, 
\label{eq:nfa41}
\end{equation} 
where $\hat{L}d^3{\bf x}_0=\tilde{L}\ d^3{\bf x}$ implies 
$\tilde{L}=\hat{L}/J$. 
Using $\hat{L}$ from (\ref{eq:4.37}) we obtain:
\begin{align}
\tilde{L}=&\hat{V}^{x^i} E_{x^i}\left(\ell^0\right)
+\mu^1\nabla{\bf\cdot}\left(\rho\hat{V}^{\bf x}\right)
+\mu^2 \hat{V}^{\bf x} {\bf\cdot}\nabla S\nonumber\\
&+\boldsymbol{\mu}^3{\bf\cdot}\nabla\times\left(\hat{V}^{\bf x}
\times {\bf B}\right)
+\boldsymbol{\mu}^4{\bf\cdot}\rho\left(\frac{d\hat{V}^{\bf x}}{dt}
-\hat{V}^{\bf x}{\bf\cdot}\nabla {\bf u}\right), \label{eq:nfa42}
\end{align}
where
\begin{equation}
\mu^1=\frac{1}{2} u^2-h,\quad \mu^2=-\rho T,
\quad \boldsymbol{\mu}^3=\frac{\bf B}{\mu_0}, 
\quad \boldsymbol{\mu}^4=-{\bf u}. \label{eq:nfa43}
\end{equation}
For the fluid relabelling symmetry 
$\hat{V}^{\bf x}=\nabla\psi\times\nabla\chi/\rho$ in (\ref{eq:nfa14}), 
Noether's second theorem gives the generalized Bianchi identity:
\begin{equation}
\frac{\delta\hat{J}}{\delta\chi}=\nabla{\bf\cdot}
\left[\nabla\psi\times\left(\frac{1}{\rho} \frac{\delta\hat{J}}{\delta\hat{V}^{\bf x}}\right)\right]=0. \label{eq:nfa44}
\end{equation}
From (\ref{eq:nfa42}) we obtain:
\begin{align}
\frac{\delta\hat{J}}{\delta\hat{V}^{\bf x}}
=&\frac{E_{\bf x}\left(\ell^0\right)}{J}-\rho\nabla\mu^1+\mu^2\nabla S
+{\bf B}\times\left(\nabla\times\boldsymbol{\mu}^3\right)\nonumber\\
&-\derv{t}\left(\rho\boldsymbol{\mu}^4\right) 
-\nabla{\bf\cdot}\left(\rho {\bf u}\otimes \boldsymbol{\mu}^4\right) 
-\rho\boldsymbol{\mu}^4{\bf\cdot}\left(\nabla{\bf u}\right)^T\nonumber\\
\equiv&\frac{E_{\bf x}\left(\ell^0\right)}{J} 
+\derv{t}\left(\rho {\bf u}\right) 
+\nabla{\bf\cdot}\left(\rho {\bf u}\otimes{\bf u} +p{\sf I}\right)
-{\bf J}\times {\bf B}, \nonumber\\
\equiv &\rho\biggl(\frac{E_{\bf x}\left(\ell^0\right)}{\rho_0}
+{\bf u}_t-{\bf u}\times\boldsymbol{\omega}+\nabla|{\bf u}|^2-{\bf F}\biggr), 
\label{eq:nfa45}
\end{align}
where $\boldsymbol{\omega}=\nabla\times{\bf u}$ is the fluid vorticity 
and ${\bf F}$ is defined in (\ref{eq:nfa16}). Using (\ref{eq:nfa45}) 
in (\ref{eq:nfa44}) then gives the generalized Bianchi 
identity (\ref{eq:nfa17}). The above analysis shows how 
Proposition \ref{prop4.6} 
could also be derived using an Eulerian variational approach.
%

\section{Summary and Concluding Remarks}
There are several related methods used to derive conservation laws in 
MHD and fluid dynamics. \cite{Moiseev82}, 
\cite{Tur93} and \cite{Webb14a}  used the technique of 
Lie dragging of $0$-forms, $p$-forms ($p=1,2,3$) (see also 
\cite{Kuvshinov97}).
This technique is relatively easy to use to derive conservation laws for 
invariants (geometrical objects) which are advected with the flow. 
Another method is to use Noether's theorems 
(e..g. \cite{Noether18}, \cite{Salmon82,Salmon88}, 
\cite{Padhye96a,Padhye96b}, \cite{Padhye98}, 
\cite{Olver93}, \cite{Hydon11},  \cite{Holm98}, 
\cite{Cotter13}, \cite{Webb05,Webb14b}). If the 
 action is invariant under a Lie point symmetry, up to a pure 
divergence transformation of the Lagrangian density,  Noether's 
 first theorem gives the corresponding conservation law.  
If the symmetry involves free functions 
due to an infinite dimensional Lie pseudo-group, then 
Noether's second theorem applies. 

A new conservation law related to Ertel's theorem for MHD was derived 
in Section 4.5 which is associated with a non-field aligned fluid relabelling 
symmetry derived by \cite{Webb05}. 
The conservation law for this symmetry is given in 
(\ref{eq:nfa20}), i.e. 
\begin{equation}
\derv{t} \left(\boldsymbol{\omega}{\bf\cdot}\nabla\psi\right)
+\nabla{\bf\cdot}\left[(\boldsymbol{\omega}{\bf\cdot}\nabla\psi){\bf u}
-\left(T\nabla S+\frac{{\bf J}\times{\bf B}}{\rho}\right)
\times\nabla\psi\right]=0, \label{eq:5.1}
\end{equation}
where $\boldsymbol{\omega}=\nabla\times {\bf u}$ is the fluid vorticity, 
${\bf J}=\nabla\times{\bf B}/\mu_0$ is the MHD current density, and
 $\psi=\psi({\bf x}_0)$ is a scalar invariant advected with the flow, 
i.e. $d\psi/dt=(\partial/\partial t+{\bf u}{\bf\cdot}\nabla)\psi=0$. 
The conservation law is more general than that expected from its derivation 
using Noether's second theorem. We expect that there are more general 
solutions of the fluid relabelling symmetry determining 
equations (\ref{eq:4.35a})-(\ref{eq:4.35c}) than those used in the Noether 
theorem derivation of (\ref{eq:5.1}), which could help to resolve this issue. 
\cite{Volkov95} discuss the Lie dragged invariants of MHD and gas dynamics 
in terms of supersymmetry and the odd Buttin bracket. The relationship 
between this approach and Noether's theorems is not clear at present.

In this paper we derived the generalized potential vorticity 
conservation law (\ref{eq:5.1}) using the  
\cite{Hydon11} approach to Noether's second theorem. 
There are other approaches to conservation laws and symmetries
that have been adopted in the literature. 
\cite{Sjoberg04}  obtained nonlocal conservation laws 
in 1D ideal fluid mechanics by means of potential symmetries 
of the equations (see also \cite{Webb09b}). 
A combination of scaling symmetries in MHD and 
fluid dynamics can also give rise to conservation laws 
(e.g. \cite{Webb07,Webb09b}).
 \cite{Bluman10} and \cite{Anco02} use a method
to derive conservation laws that sidesteps the use of Noether's theorem and 
applies to systems of equations that do not admit a variational 
principle. However, \cite{Ibragimov07,Ibragimov11}   investigated the augmented system, 
consisting of the original system of equations plus the adjoint system. 
Noether's theorem can then be applied to the augmented system to obtain conservation 
laws. 

\section*{Acknowledgements}
GMW acknowledges stimulating discussions with Darryl Holm
on fluid relabelling symmetries,
Lie point symmetries and Euler Poincar\'e equations for
fluids and plasmas. We acknowledge discussions on Noether's theorems
with Kesh Govinder and J. F. McKenzie. The work of RLM was supported
 by the National Research Foundation (NRF) of South Africa.
Any opinions, findings, and conclusions or
recommendations expressed in this material are those of the authors,
and therefore, the NRF does not accept any liability in regard
thereto.
 The work of GMW was carried out in part, during a visit to
the University of Kwa-Zulu Natal, Durban Westville in October 2013.
The work of GMW was supported in part by the South African,
National Institute of Theoretical Physics, and also by
CSPAR travel funds and the generosity of the director of CSPAR, G. P. Zank.


\appendix
\section*{Appendix A}
\setcounter{section}{1}
In this appendix we provide an independent derivation of the Lie 
invariance condition (\ref{eq:4.35aa}) for the MHD action:
\beqn
{\cal A}=\int \ell\ d^3x\ dt, \label{eq:A1}
\eeqn
to admit a divergence symmetry, 
where
\begin{equation}
\ell=\frac{1}{2}\rho u^2-\varepsilon(\rho,S)-\frac{B^2}{\mu}
-\rho \Phi({\bf x}), \label{eq:A2}
\eeqn
is the Eulerian MHD Lagrangian. We use Eulerian variations of the form:
\beqn
{\bf x}'={\bf x}+\epsilon {\hat V}^{\bf x}, \quad \ell'=\ell+\epsilon D_{x^\alpha} \Lambda^\alpha, \label{eq:A3}
\eeqn
where
\beqn
{\hat V}^{x^i}=V^{x^i}-V^{x_0^s} x_{is}-V^t x_{it}. \label{eq:A4}
\eeqn 

From \cite{Holm98}, \cite{Cotter13} and \cite{Webb14b} 
the Eulerian variation of the fluid velocity $\delta {\bf u}$ is given by:
\beqn
\delta{\bf u}=\left(\derv{t}+{\cal L}_{\bf u}\right)\hat{V}^{\bf x}
=\deriv{{\hat V}^{\bf x}}{t}
+\left[ {\bf u},{\hat V}^{\bf x}\right]_L
\equiv \frac{d{\hat V}^{\bf x}}{dt}
-{\hat V}^{\bf x}{\bf\cdot}\nabla {\bf u}, 
\label{eq:A5}
\eeqn
where ${\cal L}_{\bf u}$ denotes the Lie derivative with respect 
to the vector field ${\bf u}$, i.e. ${\cal L}_{\bf u}({\hat V}^{\bf x})
=[{\bf u},{\hat V}^{\bf x}]={\bf u}{\hat V}^{\bf x}
-{\hat V}^{\bf x} {\bf u}\equiv {\rm ad}_{\bf u}({\hat V}^{\bf x})_L$ 
is the left Lie bracket of the 
vector fields ${\bf u}$ and ${\hat V}^{\bf x}$. 
The variations of the other physical variables in (\ref{eq:A2}) (denoted 
generically by $\psi$) are given by 
the Lie derivative formula:
\beqn
\delta\psi=-{\cal L}_{{\hat V}^{\bf x}}(\psi). \label{eq:A6}
\eeqn

Using the Lie derivative formulae:
\begin{align}
\delta\left(\rho d^3x\right)
=&-{\cal L}_{{\hat V}^{\bf x}}
\left(\rho d^3x\right)
=-\nabla{\bf\cdot}\left(\rho{\hat V}^{\bf x}\right) d^3x, \quad
\delta S=-{\cal L}_{{\hat V}^{\bf x}}(S)
=-{{\hat V}^{\bf x}}{\bf\cdot}\nabla S, \nonumber\\
\delta \left({\bf B}{\bf\cdot} d{\bf S}\right)
=&-{\cal L}_{{\hat V}^{\bf x}}
\left({\bf B} {\bf\cdot} d{\bf S}\right)
=\left[\nabla\times\left({\hat V}^{\bf x}\times {\bf B}\right)
-{\hat V}^{\bf x} \nabla{\bf\cdot}{\bf B}
\right]{\bf\cdot}d {\bf S}, \label{eq:A7}
\end{align}
we obtain:
\begin{align}
\delta\rho=&-\nabla{\bf\cdot}\left(\rho {\hat V}^{\bf x}\right),
\quad \delta S=-{\hat V}^{\bf x}{\bf\cdot}\nabla S, \nonumber\\
\delta{\bf B}=&\nabla\times\left({\hat V}^{\bf x}\times{\bf B}\right)-
{\hat V}^{\bf x}\nabla{\bf\cdot}{\bf B}. \label{eq:A8}
\end{align}
Note that $\delta\rho$, $\delta S$ and $\delta {\bf B}$ are Eulerian 
variations in which ${\bf x}$ is held constant (cf. \cite{Newcomb62}, 
\cite{Webb05}).

The variation of the Lagrangian $\ell$ in (\ref{eq:A2}) is given by:
\beqn
\delta\ell=
\frac{\delta\ell}{\delta{\bf u}}{\bf\cdot}\delta{\bf u}
+\frac{\delta\ell}{\delta\rho}\delta\rho 
+\frac{\delta\ell}{\delta{\bf B}}{\bf\cdot}\delta{\bf B}
+\frac{\delta\ell}{\delta S}\delta S+D_{x^\alpha}\Lambda^\alpha. \label{eq:A9}
\eeqn
Using the Lagrangian density (\ref{eq:A2}) we obtain:
\begin{align}
\frac{\delta\ell}{\delta\rho}=&\frac{1}{2} u^2-\varepsilon_\rho-\Phi
\equiv \frac{1}{2} u^2-h-\Phi, 
\quad \frac{\delta\ell}{\delta{\bf u}}=\rho {\bf u}, 
\nonumber\\
\frac{\delta\ell}{\delta S}=&-\varepsilon_S=-\rho T, 
\quad \frac{\delta\ell}{\delta{\bf B}}=\frac{\bf B}{\mu_0}. 
\label{eq:A10}
\end{align}
Using (\ref{eq:A5}), (\ref{eq:A8}) and (\ref{eq:A10}) in (\ref{eq:A9}) gives:
\begin{align}
\delta\ell=&\left(h+\Phi-\frac{1}{2}\rho u^2\right) 
\nabla{\bf\cdot}\left(\rho {\hat V}^{\bf x}\right) 
+\rho T{\hat V}^{\bf x}{\bf\cdot}\nabla S\nonumber\\
&+\frac{\bf B}{\mu_0}{\bf\cdot}
\left[-\nabla\times\left({\hat V}^{\bf x}
\times{\bf B}\right)+{\hat V}^{\bf x} \nabla{\bf\cdot}{\bf B}\right]\nonumber\\
&+\rho{\bf u}{\bf\cdot}\left[\frac{d{\hat V}^{\bf x}}{dt}-{\hat V}^{\bf x}{\bf\cdot}\nabla{\bf u}\right]
+D_{x^\alpha}\Lambda^\alpha. \label{eq:A11}
\end{align}
Setting $\delta {\cal A}=0$ (i.e. $\delta\ell=0$) gives the divergence symmetry 
condition for the action given in (\ref{eq:4.35aa}).

\def\newblock{\hskip .11em plus .33em minus .07em}

\end{document}